\documentclass[runningheads]{llncs}
\usepackage{graphicx}
\usepackage{times}
\usepackage{xcolor}
\usepackage{soul}
\usepackage[utf8]{inputenc}
\usepackage[small]{caption}
\usepackage{footnote}
\usepackage{amsmath}
\usepackage{booktabs}
\usepackage{multirow}
\usepackage{graphicx}
\usepackage{caption}
\usepackage{array}
\usepackage{enumitem}
\usepackage{nicefrac}

\usepackage{amsthm}

\usepackage{hyperref}

\usepackage[ruled,linesnumbered,vlined]{algorithm2e}

\usepackage{etoolbox}

\BeforeBeginEnvironment{figure}{\vskip-2.5ex}
\AfterEndEnvironment{figure}{\vskip-2.5ex}

\BeforeBeginEnvironment{table}{\vskip-3ex}
\AfterEndEnvironment{table}{\vskip-2.5ex}

\newtheorem{Example}{Example}
\newtheorem{Proposition}{Proposition}
\newtheorem{Definition}{Definition}

\newtheorem{Theorem}{Theorem}

\begin{document}
\title{Approximating Voting Rules from Truncated Ballots}
%

%
\authorrunning{Ayadi \textit{et al.}}

 \author{Manel Ayadi\inst{1,2} \and
 Nahla Ben Amor\inst{1} \and
 Jérôme Lang \inst{2}}
 \institute{LARODEC, Institut Sup\'erieur de Gestion, Universit\'e de Tunis, Tunis, Tunisia
 \and
 LAMSADE, CNRS, Universit\'e Paris-Dauphine, Paris, France}

\maketitle              
\begin{abstract}
Classical voting rules assume that  ballots are complete preference orders over candidates. However, when the number of candidates is large enough, it is too costly to ask the voters to rank all candidates. We suggest to fix a rank $k$, to ask all voters to specify their best $k$ candidates, and then to consider ``top-$k$ approximations" of rules, which take only into account the $top$-$k$ candidates of each ballot. 
  We consider two measures of the quality of the approximation: the probability of selecting the same winner as the original rule, and the score ratio. We do a worst-case study (for the latter measure only), and for both measures, an average-case study and a study from real data sets. 

\end{abstract}

\section{Introduction}
The input of a voting rule is usually a collection of complete rankings over candidates (although there are exceptions, such as approval voting).
However, requiring a voter to provide a complete ranking over the whole set of candidates can be difficult and costly in terms of time and cognitive effort.
We suggest to ask voters to report only their $top$-$k$ candidates, for some (small) fixed value of $k$ (the obtained ballots are then said to be {\em top-$k$}). Not only it saves communication effort, but it is also often easier for a voter to find out the top part of their preference relation 
than the bottom part. However, this raises the issue of how usual voting rules should be adapted to top-$k$ ballots. 
Reporting top-{\em $k$} ballots is a specific form of {\em voting with incomplete preferences}, and is highly related to
\textit{vote elicitation}. Work on these topics is reviewed in the recent handbook chapter \cite{BoutilierR16}.
Existing work on truncated ballots can be classified into two classes according to the type of interaction with the voters:

\paragraph{\bf (i) Interactive elicitation } ~~~~~~~

An interactive elicitation protocol asks voters to expand their truncated ballots in an incremental way, until the outcome of the vote is eventually determined. This line of research starts with  Kalech et al. \cite{kalech2011practical} who start by top-1 ballots, then top-2, etc., until there is sufficient information for knowing the winner. Lu and Boutilier \cite{lu2011vote,lu2011robust} propose an incremental elicitation process using \textit{minimax regret} to predict the correct winner given partial information.
A more general incremental elicitation framework, with more types of elicitation questions, is cost-effective elicitation \cite{ZhaoLWKMSX18}.
Naamani Dery {\em et al.} \cite{dery2014reaching} present two elicitation algorithms for finding a winner with little communication between voters.

\paragraph{\bf (ii) Non-interactive elicitation} ~~~~~~~

The central authority elicits the top-$k$ ballots at once, for a fixed value of $k$, and outputs a winner
without requiring voters to provide extra information. A possibility consists in computing possible winners given these truncated ballots: this is the path followed by Baumeister et al. \cite{baumeister2012campaigns} (who also consider double-truncated ballots where each voter ranks some of her top and bottom candidates). Another possibility -- which is the one follow -- consists in generalizing the definition of a voting rule so that it takes truncated ballots as input.
In this line,
Oren {\em et al.} \cite{oren2013efficient} analyze $top$-$k$ voting by assessing the values of $k$ needed to ensure the true winner is found with high probability for specific preference distributions. Skowron {\em et al.} \cite{skowron2015achieving} use $top$-$k$ voting as a way to approximate some multiwinner rules. Filmus and Oren \cite{filmus2014efficient} study the performance of top-$k$ voting under the impartial culture distribution for the Borda, Harmonic and Copeland rules. 
They assess the values of $k$ needed to find the true winner  with high probability, and
they report on numerical experiments that show that under the impartial culture, top-$k$ ballots for reasonable small values of $k$ give accurate results. 

Bentert and Skowron \cite{bentert2019comparing} focus on top-$k$ approximations of voting rules that are defined via the maximization of a score (positional scoring rules and maximin). They evaluate the quality of the approximation of a voting rule by a top-$k$ rule by the worst-case ratio between the scores, with respect to the original profile, of the winner of the original rule and the winner of the approximate rule. They identify the top-$k$ rules that best approximate positional scoring rules (we give more details in Section \ref{ratio}). 
Their theoretical analysis is completed by numerical experiments using  profiles generated from different distributions over preferences: they show that for the Borda rule a small value of $k$ is needed to achieve a high approximation guarantee while maximin needs more information from a sufficiently many voters to determine the winner.

Ayadi \textit{et al.} \cite{ayadi2019single} evaluate the extent to which STV with $top$-$k$ ballots approximates STV with full information. They show that for small $k$, $top$-$k$ ballots are enough to identify the correct winner quite frequently, especially for data taken from real elections. Finally, the recognition of singled-peaked $top$-$k$ profiles is studied in \cite{Lackner14} while the computational issues of manipulating rules with $top$-$k$ profiles is addressed in \cite{NarodytskaW14}. \medskip

Our contribution concerns non-interactive elicitation. We adapt different voting rules to truncated ballots: we define approximations of voting rules which take as input the $top$-$k$ candidates of each ballot. The question is then, \textit{are these approximations good predictors of the original rule?} We answer this question by considering two measures: the probability that the approximate rule selects the `true' winner, and the ratio between the scores (for the original rule) of the true winner and the winner of the approximate rule. For the latter measure we give a worst-case theoretical analysis. For both measures we give an empirical study, based on randomly generated profiles and on real-world data. Our findings are that for several common voting rules, both for randomly generated profiles and real data, a very small $k$ suffices.

Our research can be seen as a continuation of Filmus and Oren \cite{filmus2014efficient}. We go further on several points: we consider more voting rules; beyond impartial culture, we consider a large scope of  distributions;
we study score distortion; and we include experiments using real-world data sets. 
 Our work is also closely related to \cite{bentert2019comparing}, who have obtained related results independently (see Sections \ref{accuracy} and \ref{ratio} for a discussion).\smallskip

Our interpretation of top-$k$ ballots is {\em epistemic}: the central authority in charge of collecting the votes and computing the outcome ignores the voters' preferences below the $top$-$k$ candidates of each voter, and has to cope with it as much as possible. Voters may very well have a complete preference order in their head (although it does not need to be the case), but they will simply not be asked to report it.
\smallskip
 
Section 2 gives some background. Section 3 defines top-$k$ approximations of different voting rules. Section 4 analyses empirically the probability that approximate rules select the true winner. Section 5 analyses score distortion, theoretically and empirically. 
\section{Preliminaries}\label{sec:preliminaries}

An election is a triple $E=\left\langle N, A, P\right\rangle$ where: $N=\lbrace 1,...,n \rbrace$ is the set of \textit{voters}, $A$ is the set of \textit{candidates}, with $\vert A\vert = m$; and  $P =(\succ_{1},...,\succ_{n})$ is the \textit{preference profile} of voters in $N$, where for each $i$, $\succ_{i} \in P$
is a linear order over $A$. ${\cal P}_m$ is the set of all profiles over $m$ alternatives (for varying $n$).  

Given a profile $P$, $N_{P}(a, b) = \# \left\{i, a \succ_{i} b\right\}$ is the number of voters who prefer $a$ to $b$ in $P$. The \textit{majority graph} $M(P)$ is the graph whose set of vertices is the set of the candidates $A$ and in which for all $a,b \in A$, there is a directed edge from $a$ to $b$ (denoted by $a \rightarrow b$) in $M(P)$ if $N_{p}(a,b) > \frac{n}{2}$.

A resolute voting rule is a function $\large f: E \to A$. Resolute rules are typically obtained from composing an irresolute rule (mapping an election into an non-empty subset of candidates, called co-winners) with a tie-breaking mechanism. 

A {\em positional scoring rule} ({\em PSR}) $f^s$ is defined by a non-negative vector $\vec{s}=\left(s_{1},...,s_{m}\right)$ 
such that $s_{1}\geq...\geq s_{m}$ and $s_{1} > 0$. Each candidate receives $s_{j}$ points from each voter $i$ who ranks her in the $j^{th}$ position, and the score of a candidate is the total number of points she receives from all voters i.e. $S(x)=\sum^{n}_{i=1}s_{j}$. The winner is the candidate with highest total score. Examples of scoring rules are the Borda and Harmonic rules, with $s_{Borda}=(m-1, m-2,\dots, 0)$ 
and $s_{Harmonic}=(1, \nicefrac{1}{2},\dots, \nicefrac{1}{m})$. 

We now define three \textit{pairwise comparison rules}. 

The {\em Copeland} rule outputs the candidate maximizing the {\em Copeland score}, where the Copeland score of $x$ is the number of candidates $y$ with $x \rightarrow y$ in $M(P)$, plus half the number of candidates $y \neq x$ with no edge between $x$ and $y$ in $M(P)$. 

The \textit{Ranked Pairs} (RP) rule proceeds by ranking all pairs of candidates $(x,y)$ according to $N_P(x,y)$ (using tie-breaking when necessary); starting from an empty graph over $A$, it then considers all pairs in the described order and includes a pair in the graph if and only if it does not create a cycle in it. At the end of the process, the graph is a complete ranking, whose top element is the winner.

The \textit{maximin} rule outputs the candidates that maximize
$min_{ x \in A \left( x \neq a \right)} N_P \left( a,x \right)$.

For the experiments using randomly generated profiles, we use 
the \textit{Mallows $\phi$-model} \cite{mallows1957non}. It is a (realistic) family of distributions over rankings, parametrized by a modal or reference ranking $\sigma$ and a dispersion parameter $\phi \in \left[ 0, 1\right]$: $P \left( r; \sigma, \phi \right) = {\frac{1}{Z}} \phi^{d \left(r,\sigma \right)}$, where $r$ is any ranking, $d$ is the Kendall tau distance and $Z = \sum_{r'} \phi^{d \left(r,\sigma \right)}= 1 \cdot \left(1+\phi \right) \cdot \left(1+\phi+\phi^{2}\right)\cdot...\cdot \left(1+...+\phi^{m - 1} \right)$ is a normalization constant. With small values of $\phi$, the mass is concentrated around $\sigma$, while $\phi = 1$ gives the uniform distribution \textit{Impartial Culture (IC)}, where all profiles are equiprobable.

\section{Approximating Voting Rules from Truncated Ballots}\label{sec:approximating}

Given $k \in \left\{1,...,m-1\right\}$, a {\em top-$k$} election is a triple $E'=\left\langle N, A, R\right\rangle$ where $N$ and $A$ are as before, and 
$R =(\succ_{1}^{k},...,\succ_{n}^{k})$, where each $\succ_{i}^{k}$ is a ranking of $k$ out of $m$ candidates in $A$. $R$ is called a 
{\em top-$k$ profile}.
If $P$ is a complete profile, $\succ_i^k$ is the top-$k$ truncation of $\succ_i$ (i.e., the best $k$ candidates, ranked as in  $\succ_i$), and $P_k = (\succ_{1}^k,...,\succ_{n}^k)$ is the top-$k$-profile induced from $P$ and $k$.
A \textit{top-$k$ (resolute) voting rule} is a function $f_k$ that maps each $top$-$k$ election $E'$ to a candidate in $A$.
 We sometimes apply a top-$k$ rule to a complete profile, with $f_k(P) = f_k(P_k)$. 
 We now define several $top$-$k$ rules.

\subsection{Borda and Positional Scoring Rules}

\begin{Definition}
A $top$-$k$ PSR $f^s_k$ is defined by a scoring vector 
$s=\left(s_{1},s_{2}\dots, s_{k},s^{*}\right)$ such that $s_{1}\geq s_{2}\geq...\geq s_{k}\geq s^{*} \geq 0$ and $s_{1} > s^{*}$. 
Each candidate in a $top$-$k$ vote receives $s_{j}$ points from each voter $i$ who ranks her in the $j^{th}$ position. A non-ranked candidate gets $s^{*}$ points. 
The winner is the candidate with highest total score. 
\end{Definition}

When starting from a specific {\em PSR} for complete ballots, defined by scoring vector $s=\left(s_{1}, \ldots, s_{m}\right)$, two choices of $s^*$ particularly make sense: 
\begin{itemize}
\item zero score: $s^* = 0$
\item average score: $s^* = \frac{1}{m-k}\left(s_{k+1} + \ldots + s_m\right)$
\end{itemize}
We denote the corresponding approximate rules as $f^0_k$ and $f^{av}_k$. $Borda^{av}_k$ is known under the name  \textit{average score modified Borda Count} \cite{cullinan2014borda,Grandi:2016:BCC:2989792.2989797},
while $Borda^{0}_k$ is known under the name \textit{modified Borda Count} \cite{emerson1994politics}). In the experiments we report only on $Borda^{av}_k$, as $Borda^{0}_k$ gives very similar results.


Young \cite{young1975social} characterized positional scoring rules by these four properties, which we describe informally (for resolute rules): 
\begin{itemize}
\item \textit{Neutrality}: all candidates are treated equally
\item \textit{Anonymity}: all voters are treated equally 
\item \textit{Reinforcement}: if $P$ and $Q$ are two profiles (on disjoint electorates) and $x$ is the winner for $P$ and the winner for $Q$, then it is also the winner for $P \cup Q$.
\item \textit{Continuity}: if $P$ and $Q$ are two profiles and $x$ is the winner for $P$ but not for $Q$, adding sufficiently many votes of $P$ to $Q$ leads to elect $x$.
\end{itemize}

$f$ is a {\em PSR} if and only if it satisfies neutrality, anonymity, reinforcement and continuity \cite{young1975social}. 

These four properties still make sense for truncated ballots. 
It is not difficult to generalize Young's result to $top$-$k$ PSR:
\begin{Theorem}
A $top$-$k$ voting rule is a $top$-$k$ {\em PSR} if and only if it satisfies 
neutrality, anonymity, reinforcement, and continuity.
\end{Theorem}

\begin{proof}
The left-to-right direction is obvious. 
For the right-to-left direction, let us first define the {\em $top$-$k$-only} property: a standard voting rule is {\em $top$-$k$-only} if for any two complete profiles $P, P'$, if $P_k = P_k'$, then $F(P) = F(P')$.  Then (1) a positional scoring rule $F$ is $top$-$k$-only if and only if $s_{k+1} = \ldots = s_m$ (if this equality is not satisfied, then it is easy to construct two profiles $P$, $P'$ such that $P_k = P_k'$ and $F(P) \neq F(P')$). Now, assume $F_k$ is a $top$-$k$ rule satisfying neutrality, anonymity, reinforcement, and continuity. Let $F$ be the standard voting rule defined by $F(P) = F_k(P_k)$. Clearly, $F$ also satisfies neutrality, anonymity, reinforcement, and continuity, and due to Young's characterization result, $F$ is a PSR, associated with some vector $(s_1, \ldots, s_m)$. Because $F$ is also $top$-$k$-only, using (1) we have  $s_{k+1} = \ldots = s_m$, therefore, $F_k$ is a $top$-$k$-PSR.
\end{proof}

\subsection{Rules Based on Pairwise Comparisons} 

Given a truncated ballot $\succ_{i}^{k}$ and two candidates $a, b \in A$, we say that $a$ dominates $b$ in $\succ_{i}^{k}$, denoted by $a >_i^k b$, if one of these two conditions holds: (1) $a$ and $b$ are listed in $\succ_{i}^{k}$, and $a \succ_{i}^{k} b$; (2) $a$ is listed in $\succ_{i}^{k}$, and $b$ is not.

For instance, for $A = \{a,b,c,d\}$, $k = 2$, and $\succ_{i}^{2} = (a \succ b)$, then $a$ dominates $b$, 
both $a$ and $b$ dominate $c$ and $d$, 
but $c$ and $d$ remain incomparable in $\succ_{i}^{2}$.
Now, the notions of pairwise comparison and majority graph are extended to $top$-$k$ truncated profiles in a straightforward way: 

\begin{Definition}
Given a top-$k$ profile $R$, $N_{R}(a, b) = \# \left\{i, a >_{i}^k b\right\}$ is the number of voters in $R$ for whom $a$ dominates $b$. The \textit{top-$k$ majority graph} $M_{k}(R)$ induced by $R$ is the graph whose set of vertices is the set of the candidates $A$ and in which there is a directed edge from $a$ to $b$ if $N_{R}(a, b) > N_{R}(b, a)$.
\end{Definition}

The top-$k$ rules $Copeland_{k}$, $Maximin_{k}$ and $RP_k$ are defined exactly as their standard counterparts, but 
starting from the top-$k$ pairwise comparisons and majority graph instead of the standard ones. Note that $f_{m-1} = f$, and (for all rules $f$ we consider) $f_1$ coincides with plurality.

\begin{Example}
Let us consider this 62-voter profile: 20 votes $a \succ d \succ c \succ b$, 10 votes $b \succ c \succ d \succ a$, 15 votes $c \succ d \succ b \succ a$ and 17 votes: $d \succ c \succ a \succ b$.
\begin{figure}[h!]
   \centering
    \includegraphics[width=6cm,height=2.5cm]{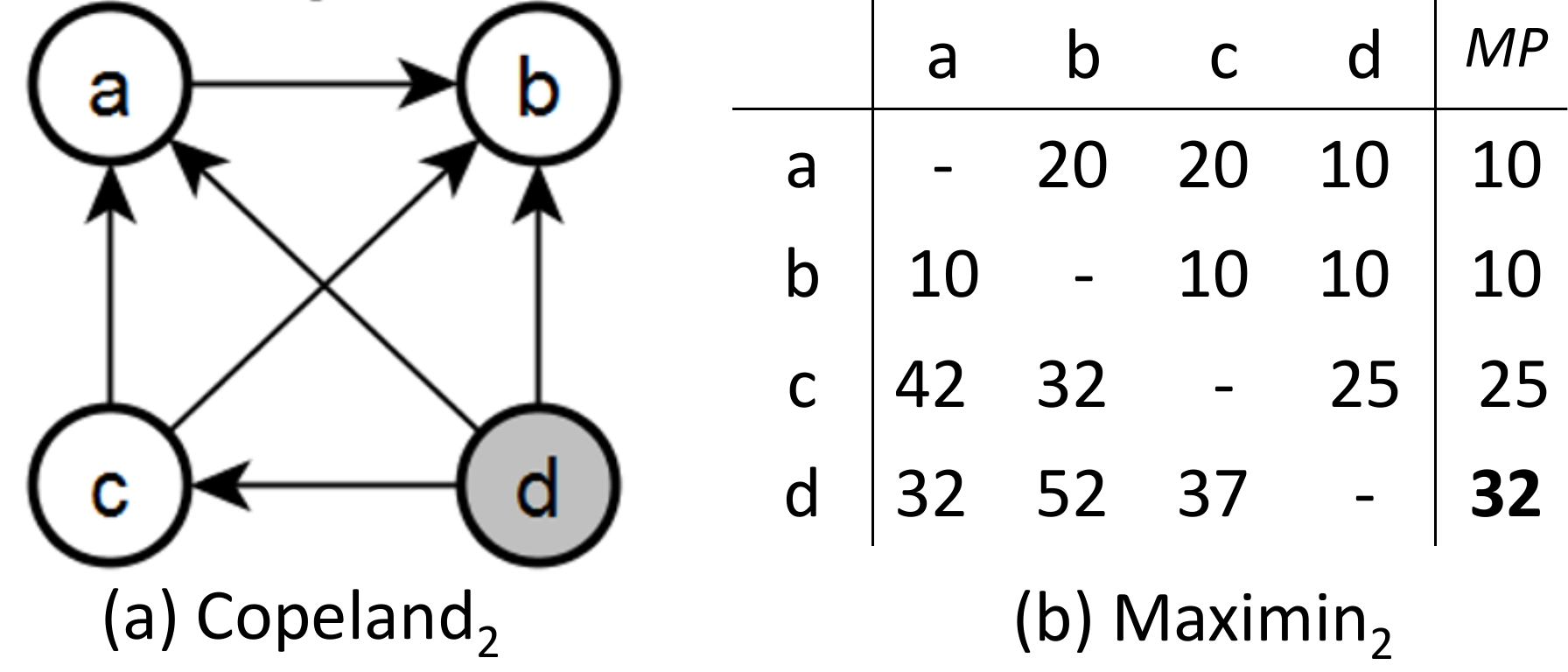}
	\caption{top-$2$ approximations of Copeland and Maximin}
\label{example}
\end{figure}
Fig.~\ref{example}~(a) shows the top-$k$ majority graph and the Copeland winner for $k = 2$, and Fig.~\ref{example}~(b) shows the top-$k$ pairwise majority matrix and the $Maximin_k$ winner for $k = 2$. In both cases, the winner for $k = 1$ (resp. $k = 3$) is $a$ (resp. $d$).
For RP, the winner under $RP_k$ for $k \in \{1,2,3\}$ is the same as the winner under $Copeland_k$ since the k-truncated majority graph does not create cycles.
	
\end{Example}
\section{Probability of Selecting the True Winner} \label{accuracy}

The first way of measuring the quality of the $top$-$k$ approximations is to determine the probability that they output the `true winner'; that is, the winner of the original voting rule, under various distributions (Subsection \ref{proba-distri}) and for real-world data (Subsection \ref{proba-real}). 
In both cases, the procedure is similar: given a voting rule $f$, we consider many profiles, and for each profile $P$ we compare $f(P)$ to $f_k(P_k)$ for each $k =\{ 1,\ldots, m-2\}$. The difference between Subsections \ref{proba-distri} and \ref{proba-real} is that in the former we randomly draw profiles according to a given distribution, and for the latter, we draw a profile by selecting $n$ votes at random in the database.
We include in our experiments $STV_k$ rule defined by Ayadi \textit{et al.} \cite{ayadi2019single}, which takes $top$-$k$ ballots as input; and we compared it to our truncated rules. $STV_k$ proceeds as follows: in each round the candidate with the smallest number of votes is eliminated (using a tie breaking when necessary), if all ranked candidates are eliminated by STV, the vote is then `exhausted' and ignored during further counting.
\subsection{Experiments Using Mallows Model}\label{proba-distri}

Here we follow the research direction initiated by Filmus and Oren \cite{filmus2014efficient}, but we consider more rules, and beyond {\em Impartial Culture} we also consider correlated distributions within the {\em Mallows} model.
For each experiment we draw 1000 random preference profiles.
In the first set of experiments, we take $m = 7$, we let $n$ and $\phi$ vary, and we measure the accuracy of the approximate rule for $k = 1$ and $k = 2$. Results are reported on Table \ref{resultsTable1}.
Note that for $k = 1$, our results can be viewed as answering the question: \textit{with which probability does the true winner with respect to the chosen rule coincide with the plurality winner?}

{\setlength{\tabcolsep}{0.1cm} 
\begin{table}[h]
	\centering
	\caption{Success rate, Mallows model: $m = 7$, varying $n$, $k$ and $\phi$.}
	\tiny{
		\begin{tabular}{|c||c|c|c|c|c ||c |c |c |c |c |}
			\hline
$\phi$& n=100 &	n=200 &	n=300 &	n=400& n=500	& n=100& n=200& n=300& n=400& n=500\\
		\hline
		
%
	
\multicolumn{6}{|c||}{$Borda^{av}_{1}$} & \multicolumn{5}{c|}{$Borda^{av}_{2}$}\\
			\hline

0.7 &0.902	&0.958&	0.986&	0.992&	\textbf{1.0}&0.951&	0.98&	0.992&	\textbf{1.0}&	1.0\\
0.8 &0.77	&0.855	&0.9&	0.94&	0.963&0.853&	0.913&	0.956&	0.972&	0.986\\
0.9 &0.588	&0.694&	0.685	&0.718&	\textbf{0.771}&0.772&	0.805&	0.827&	0.846&	0.873\\
1 &0.434&	0.445&	0.424&	0.422	&\textbf{0.397}&0.576	&0.56&	0.586&	0.598&	\textbf{0.584}\\	
	\hline
%
	%
\multicolumn{6}{|c||}{$Copeland_{1}$} & \multicolumn{5}{c|}{$Copeland_{2}$}\\
\hline

0.7 &0.908&	0.968&	0.991&	0.994&	\textbf{1.0}&0.947&	0.99&	1.0&	\textbf{1.0}&	1.0\\
0.8 &0.736&	0.847&	0.891&	0.934&	0.949&0.822&	0.904&	0.952&	0.984&	0.982\\
0.9 &0.497&	0.567&	0.655&	0.684&	\textbf{0.726}&0.62&	0.69&	0.77&	0.805&	0.838\\
1   &0.325& 0.332&	0.323&	0.343&	\textbf{0.319}&0.458&	0.432&	0.45&	0.442&	\textbf{0.425}\\
     \hline
		\multicolumn{6}{|c||}{$Maximin_{1}$} & \multicolumn{5}{c|}{$Maximin_{2}$}\\
		
		\hline

0.7 &	0.908&	0.969&	0.986&	0.99&	\textbf{1.0}&0.968&	0.991&	1.0&	\textbf{1.0}	&   1.0\\
0.8 &	0.787&	0.856&	0.915&	0.939&	0.955&0.872&	0.934&	0.961&	0.976&	0.977\\
0.9 &	0.57&	0.633&	0.691&	0.717&	\textbf{0.748}&0.735&	0.76&	0.794&	0.838&	0.869\\
1   &	0.415&	0.4&	0.423&	0.393&	\textbf{0.391}&0.52&   0.532&	0.544&	0.545&	\textbf{0.525}\\

\hline

		\multicolumn{6}{|c||}{$Harmonic_{1}$} & \multicolumn{5}{c|}{$Harmonic_{2}$}\\
		
		\hline

0.7 &	0.941 & 0.986 & 0.996 & 1.0   & \textbf{1.0}  & 0.98 & 0.992 & 1.0  & \textbf{1.0}   & 1.0\\
0.8 &	0.895 & 0.916 & 0.958 & 0.959 & 0.968& 0.958& 0.974 & 0.987& 0.988 & 0.996\\
0.9 &	0.805 & 0.808 & 0.83  & 0.866 & \textbf{0.863}& 0.895& 0.921 & 0.934& 0.939 & 0.952\\
1   &	0.725 & 0.742 & 0.74  & 0.697 & \textbf{0.737}& 0.872& 0.867 & 0.859& 0.861 & \textbf{0.859}\\

\hline

		\multicolumn{6}{|c||}{$RP_{1}$} & \multicolumn{5}{c|}{$RP_{2}$}\\
		
		\hline

0.7 &	0.926 & 0.972& 0.995& 0.995& \textbf{1.0}  & 0.963& 0.994& 1.0  & \textbf{1.0}  & 1.0 \\
0.8 &	0.778 & 0.856& 0.908& 0.939& 0.957& 0.871& 0.928& 0.967& 0.983& 0.989\\
0.9 &	0.587 & 0.64 & 0.674& 0.718& \textbf{0.749}& 0.725& 0.765& 0.777& 0.838& 0.862\\
1   &	0.426 & 0.405& 0.416& 0.375& \textbf{0.385}& 0.558& 0.524& 0.557& 0.498& \textbf{0.519}\\

\hline

		\multicolumn{6}{|c||}{$STV_{1}$} & \multicolumn{5}{c|}{$STV_{2}$}\\
		
		\hline

0.7 &	 0.907& 0.981& 0.985& 0.998& \textbf{1.0}& 0.959& 0.993&  0.997& \textbf{1.0}  &   1.0\\
0.8 &	 0.808& 0.865& 0.917& 0.918&  0.943      & 0.882& 0.933&  0.962& 0.966         &  0.974\\
0.9 &	 0.603&  0.64& 0.721& 0.729&  0.763      & 0.742& 0.776&  0.792&  0.855        &  0.846\\
1   &	  0.45& 0.464& 0.477& 0.471& 0.468       & 0.576& 0.593&  0.61 &  0.592        &  0.585\\

\hline

		\end{tabular}}
\label{resultsTable1}
\end{table}}

{\em For $k = 1$:} when $n \leq 100$ and $\phi \leq 0.7$, prediction reaches 90\% for Borda, Copeland, Maximin and STV, 92\% for RP, and 94\% for Harmonic. When $n \geq 500$, the accuracy is perfect for all rules.
For $\phi = 0.8$, the success rate decreases but results are still good with a large number of voters.
For $\phi = 0.9$ and $n = 500$, the rate reaches 86\% for Harmonic and 72\% for Copeland, with intermediate (and similar) results for Borda, Maximin and $RP$ and STV.
For the $IC$, the rate decreases dramatically when $k$ becomes small, except for Harmonic (73\% when $n=500$ against 46\% for STV, 31\% for Copeland and 40\% for the remaining rules).

{\em For $k = 2$:} the probability of selecting the true winner reaches 100\% (resp. 98\%) when $\phi \leq 0.7$ (resp. $\phi \leq 0.8$) and $n \geq 400$ (resp. $n \geq 500$). With high values of $\phi$, 
Harmonic still outperforms other rules followed by $Borda^{av}$ and STV then the other rules.
Consistently with the results obtained by Bentert and Skowron \cite{bentert2019comparing} for the IC, approximating the maximin rule is harder than position scoring rules where maximin needs more information from the voters in order to obtain high approximation guarantees.
In all cases, top-2 ballots seem to be always sufficient in practice to predict the winner with 100\% accuracy with a low value of $\phi$.


In the second set of experiments, we are interested in determining the value of $k$ needed to predict the correct winner with large elections and with high value of $\phi$. We take $k = \left\{1,...,m\right\}$, $n=2000$, $\phi=\{0.9,1\}$ and $m=20$. 
Fig.~\ref{varym} shows depicted results where 1000 random preference profiles are generated for each experiment.
Results suggest that in large elections and unless $\phi$ is very high ($\phi=0.9$), top-$k$ rules are able to identify the true winner 
when $k = 6$ (resp. $k = 8$) for Harmonic (resp. the remaining rules) out of $m=20$.
We can also observe the behavior of different truncated rules when $\phi=0.9$: the best accuracy is obtained again by Harmonic and the accuracy of all other rules are very close, which we found surprising.
When $\phi = 1$, the latter behavior changes: Harmonic still has the best results, followed by $Borda^{av}$ and STV, then the remaining rules. 
The good performance of Harmonic in all cases can be explained by the fact that the closer the scoring vector to plurality, the better the prediction. 

\begin{figure}[h]
\hspace{-0.5cm}
    \includegraphics[width=12.5cm]{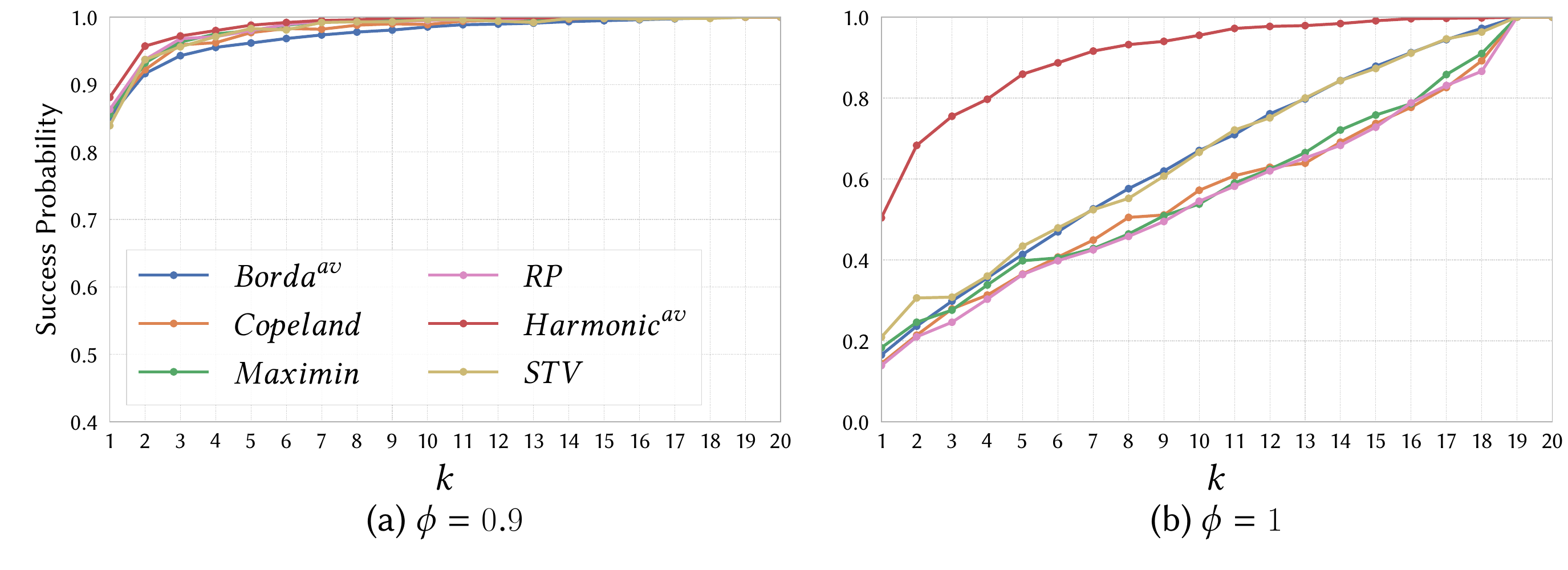}
	\caption{Success rate, Mallows model: $n = 2000$, $m =20$, varying $\phi$ and $k$.} 
\label{varym}
\end{figure}

Next, for each value of $n \in \{1000, 2000\}$, $\phi \in \{.7, .8, .9, 1\}$, and $m \in \{7, 10, 15, 20\}$, we generated 1000 random profiles, and for each of our rules, we determined the minimal value $k$ (as a function of $m$) such that the winner is correctly determined from the top-$k$ votes for all generated profiles. The results for $Borda^{av}$ are:
\begin{itemize}
\item for $\phi = 0.7$, $k=1$ is always sufficient, whatever $m$.
\item for $\phi = 0.8$, $k=2$ (resp. $k = 1$) is always sufficient for $n = 1000$ (resp. $n = 2000$), whatever the value of $m$.
\item for $\phi = 0.9$, we observe that the minimal value of $k$ such that the correct winner is always correctly predicted is around $\frac{7}{10}m$ (for $n = 1000)$ and $\frac{2}{5}m$ (for $n = 2000)$.
\item for $\phi = 1$, the minimal value of $k$ is $m-1$: we always find a generated profile for which we get an incorrect result if the profile is not complete.
\end{itemize}

The results for Copeland, maximin, RP and STV are similar to those for Borda. For Harmonic, we observe that $k = 1$ is always sufficient for $\phi \leq 0.8$ and $n=2000$, and that for $\phi =0.9$ (resp. $\phi =1$), the value of $k$ needed is around $\frac{1}{3}m$ (resp. $\frac{2}{3}m$).


In order to see how our approximations behave with small number of voters and a high dispersion parameter, we 
take $k = \left\{1,...,m\right\}$, $n=15$, $m = 7$, and $\phi \in \{0.9,1\}$. The results are on Fig.~\ref{results2}. 
The worst performance is obtained with Copeland, while the other rules perform more or less equally well. 
These results are consistent with the results obtained by Skowron et al. \cite{skowron2015achieving} for multiwinner rules: elections with few voters and high dispersion appear to be the worst-case scenario for predicting the correct winner using top-truncated ballots.
For Harmonic, even with few voters, winner prediction is almost perfect when $k=4$ and $m = 7$.
\begin{figure}[h]
   \centering
    \includegraphics[width=10cm]{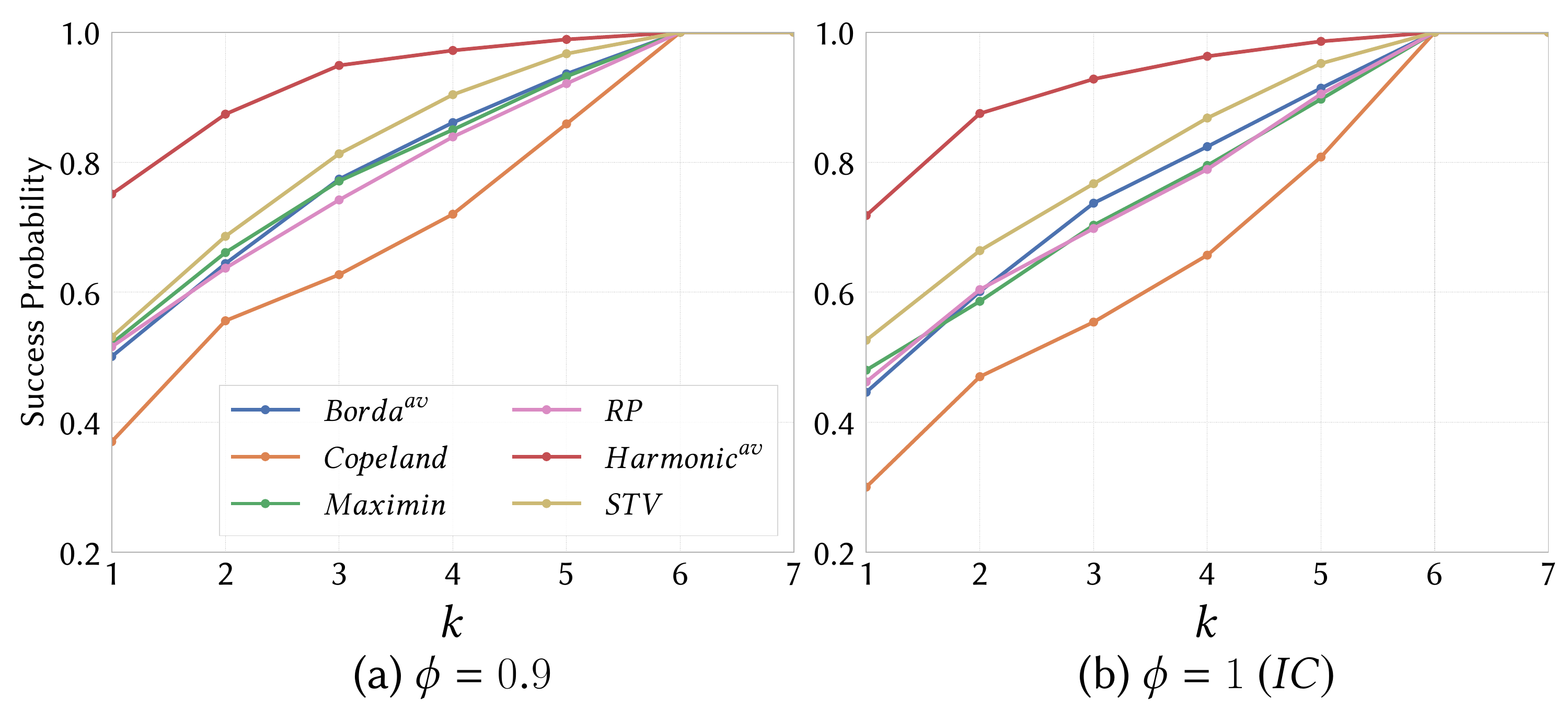}
	\caption{Success rate, Mallows model: $m = 7$, $n = 15$, varying $\phi$ and $k$.}
\label{results2}
\end{figure}


\vspace{-0.5cm}
\subsection{Experiments Using Real Data Sets}\label{proba-real}
We now consider real data set from \textit{Preflib} \cite{mattei2013preflib}: 
2002 election for Dublin North constituency with 12 candidates and 3662 voters.
We consider data with samples of $n^*$ voters among $n$ ($n^* < n$), starting by $n^* =10$ and increment $n^*$ in steps of $10$. In each experiment, 1000 random profiles are constructed with $n^*$ voters; then we consider the top-$k$ ballots obtained from these profiles, with $k = \{1,2,3\}$, and we compute the frequency with which we select the true winner. 
Fig.~\ref{Real0} shows results for Dublin with small elections ($n^* = \{10,...,100\}$) while Fig.~~\ref{Real1} presents results for large elections ($n^* = \{100,...,2000\}$). Arrows indicate the number of voters from which the prediction is perfect.

Consistently with the results of Fig.~\ref{results2}, for small elections; the success rate is low when $k$ is too small, except for Harmonic where it 
gives the best performance followed by STV (especially when $n^{*} <60$) then the remaining rules, 
e.g. For Harmonic (resp. STV), 92\% (resp. 82\%) accuracy is reached with $k = 3$, $m = 12$ and $n^{*} = 50$ against around 75\% for the remaining rules.
\begin{figure}[h]
  \hspace{-0.5cm}
    \includegraphics[width=13cm]{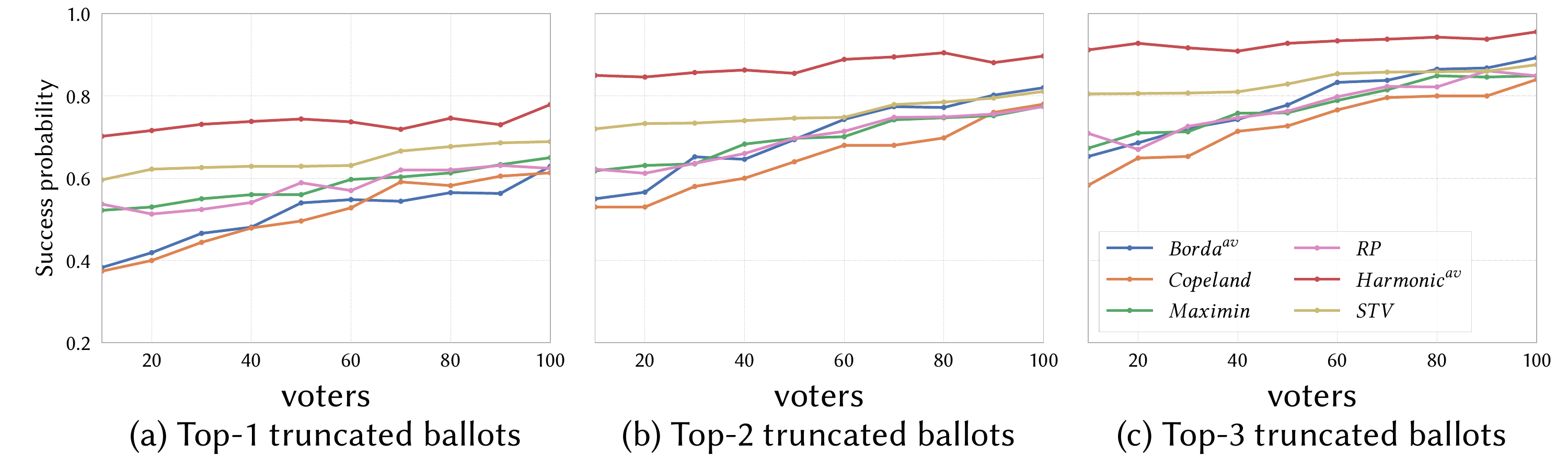}
	\caption{Success rate, Dublin, varying $k$; $n^* = \{10,\dots,100\}$.} 
\label{Real0}
\end{figure}

\begin{figure}[h]
   \hspace{-0.5cm}
    \includegraphics[width=13cm]{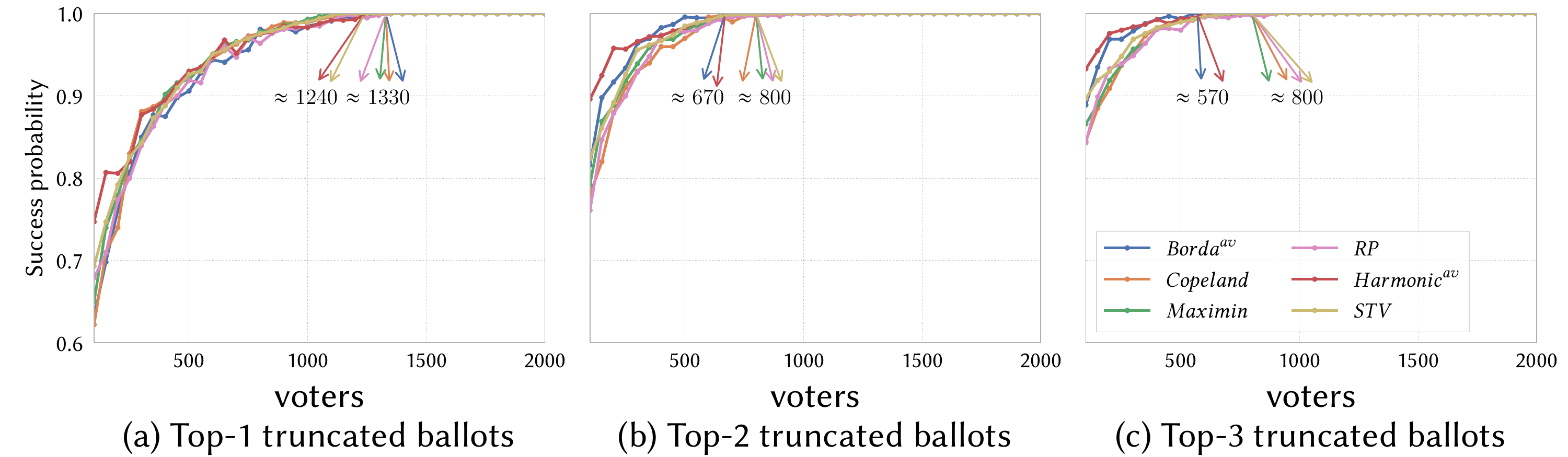}
	\caption{Success rate, Dublin, varying $k$; $n^* = \{100,\dots,2000\}$.} 
\label{Real1}
\end{figure}

For large elections, 
when $k=1$, the different approximations exhibit almost the same behavior except Harmonic, that performs better especially with few voters.
Obviously, increasing the value of $k$ leads to a decrease in the number of voters needed for correct winner selection. 
In general, the different approximations needs a sufficient number of voters to converge to the correct prediction. Scoring rules tend to require less voters.
%

\vspace{-0.5cm}
\section{Measuring the Approximation Ratio}\label{ratio}

\subsection{Worst Case Study}\label{worst}

In order to measure the quality of approximate voting rules whose definition is based on score maximization, a classical method consists in computing the worst-case approximation ratio between the scores (for the original rule) of the `true' winner and of the winner of the approximate rule. Using worst-case score ratios is classical: they are defined for measuring the quality of approximate voting rules  \cite{caragiannis2014socially,ServiceA12a},
for defining the price of anarchy of a voting rule \cite{branzei2013bad} or for measuring the distortion of a voting rule \cite{BoutilierCHLPS15}.

Worst-case score ratios particularly make sense if the score of a candidate is meaningful beyond its use for determining the winner. This is definitely the case for Borda, as the Borda count is often seen as a measure of social welfare (see \cite{Aspremont02}). 
This worst-case score ratio is called the {\em price of top-$k$ truncation}.
\begin{Definition}
Let $f$ be a voting rule defined as the maximization of a score $S$, and $f_k$ a top-$k$ approximation of $f$. The {\em price of top-$k$-truncation} for $f$, $f_k$, $m$, and $k$, is defined as: $R(f,f_k,m,k) = \max_{P \in {\cal P}_m} \frac{S(f(P))}{S(f_k(P_k))}$.
\end{Definition}




\paragraph{\bf Positional Scoring Rules }~~~~~~~~~~~~~~~~~~~~~~~~~~~~

Let $f^s$ be a positional scoring rule defined with scoring vector $s$. 
Assume the tie-breaking priority favors $x_1$.
Let $f_k^{\bar{s}}$ be a top-$k$ approximation of $f^s$, associated with vector $\bar{s} =(s_1, \ldots, s_k,s^*)$, with the same tie-breaking priority. Let $s' = (s_1-s^*, \ldots, s_k-s^*, 0) = (s'_1, \ldots, s'_k, 0)$,  i.e., $s'_i = s_i - s^*$ for $i = 1, \ldots, k$.  Obviously, $f_k^{\bar{s}} = f_k^{s'}$.  For instance, if $f^{\bar{s}}$ is the average-score approximation of the Borda rule, then $\bar{s} = (m-1, \ldots, m-k, \frac{m-k-1}{2})$ and $s' = (m-1-\frac{m-k-1}{2}, \ldots, m-k-\frac{m-k-1}{2}, 0)$. 

Let $S(x,P)$ be the score of $x$ for $P$ under $f^s$ and $S'_k(x,P_k)$ be the score of $x$ for $P_k$ under $f_k^{s'}$.
From now on when we write scores we omit $P$ and $P_k$, i.e., we write $S(x)$ instead of $S(x,P)$, $S'_k(x)$ instead of $S'_k(x,P_k)$  etc. In the rest of Subsection \ref{worst} we assume $k \geq 2$. Let $x_1 = f^{s'}_k(P_k)$ and $x_2 = f^s(P)$. 
\begin{lemma}\label{upper}
$R(f^s,f^{s'}_{k},m,k) \leq 1 - \frac{s_{k+1}}{s'_1} + \left(1+\frac{s^*}{s'_1}\right) \frac{m s_{k+1}}{s_1' + \ldots + s'_k}$
\end{lemma}

\begin{proof}
The total number of points given to candidates under $f^{s'}_k$ is $n(s'_1 + \ldots + s'_k)$, therefore $S'_k(x_1) 
\geq \frac{n}{m}(s'_1 + \ldots + s'_k)$.

Let us write $S(x_2) = S_{1 \to k}(x_2)+ S_{k+1\rightarrow m}(x_2)$, where $S_{1 \to k}(x_2)$ (resp. $S_{k+1\rightarrow m}(x_2)$) is the number of points that $x_2$ gets from the top $k$ (resp. bottom $m-k$) positions of the ballots in $P$.
Let $\gamma$ be the number of ballots in which $x_2$ is not in the top $k$ positions. Then 
$S_{k+1\rightarrow m}(x_2) \leq \gamma s_{k+1}$. 

As $x_2$ appears in at least $\frac{S'_k(x_2)}{s'_1}$ top-$k$ ballots, we have $\gamma \leq n - \frac{S'_k(x_2)}{s'_1}$. Moreover we have $S(x_1) \geq S_{1 \to k}(x_1) = S'_k(x_1)+ n s^* \geq S'_k(x_2)+ n s^* = S_{1 \to k}(x_2)$. Now, 
$$\begin{array}{lll}
S(x_2) & \leq &  S_{1 \to k}(x_2) + \left(n - \frac{S_k'(x_2)}{s'_1}\right)s_{k+1}\\
& \leq &  S_{1 \to k}(x_2) + \left(n - \frac{S_k(x_2)-n s^*}{s'_1}\right)s_{k+1}\\
& \leq &  (1 - \frac{s_{k+1}}{s'_1})S_{1 \to k}(x_2) + ns_{k+1} + \frac{n s^* s_{k+1}}{s'_1} \\
& \leq &  (1 - \frac{s_{k+1}}{s'_1})S(x_1) + ns_{k+1} + \frac{n s^* s_{k+1}}{s'_1} 
\end{array}$$

$$\begin{array}{lll}
\frac{S(x_2)}{S(x_1)} & \leq  & 1 - \frac{s_{k+1}}{s'_1} + n s_{k+1}(1+\frac{s^*}{s'_1}) \frac{m}{n(s'_1 + \ldots + s'_k)}\\
 & \leq  & 1 - \frac{s_{k+1}}{s'_1} + s_{k+1}(1+\frac{s^*}{s'_1}) \frac{m}{s'_1 + \ldots + s'_k}
 \qedhere
\end{array}$$
\end{proof}

We now focus on the lower bound.
We build the following pathological complete profile $P$ such that:

\begin{itemize} 
\item  the winner for $P_k$ (resp. $P$) is $x_1$ (resp. $x_2$).
\item  in $P_k$, all candidates get the same number of points ($x_1$ wins thanks to tie-breaking), and $x_1$ and $x_2$ get all their points from top-1 positions.
\item  in $P$, the score of $x_1$ is minimized by ranking it last everywhere where it was not in the top $k$ positions, and  the score of $x_2$ is maximized by ranking it in position $k+1$ everywhere where it was not in the top $k$ positions.
\item $P_k$ is symmetric in $\{x_3, \ldots, x_m\}$.
\end{itemize}
Formally, $P_k$ is defined as follows:

\begin{enumerate}
\item for each ranked list $L$ (resp. $L'$) of $k-1$ (resp. $k$) candidates in $\{x_3, \ldots, x_m\}$: $\alpha$ votes $x_1 L$ and $\alpha$ votes $x_2 L$ (resp. $\beta$ votes $L'$). $\alpha$ and $\beta$ will be fixed later.
\item $\alpha$ and $\beta$ are chosen in such a way that all candidates get the same score $S'_k(.)$.
\end{enumerate}

Now, $P$ is obtained by completing $P_k$ as follows:

\begin{enumerate}
\item each top-$k$ vote  $x_1 L$ is completed into  $x_1 L x_2 -$. ``$-$'' means the remaining candidates are in an arbitrary order.
\item each top-$k$ vote $x_2 L$ is completed into  $x_2 L - x_1$.
\item each top-$k$ vote $L'$ is completed into  $L' x_2 - x_1$.
\end{enumerate}
For instance, for $m = 5$ and $k = 3$, $P$ is as follows:\medskip
\small \noindent
\hspace{-4mm}
\begin{center}
\begin{tabular}{c|c|c}
\begin{tabular}{ll}
$\alpha$ & $x_1 x_3 x_4 x_2 x_5$\\
$\alpha$ & $x_1 x_3 x_5 x_2 x_4$\\
$\alpha$ & $x_1 x_4 x_3  x_2 x_5$\\
$\alpha$ & $x_1 x_4 x_5  x_2 x_3$\\
$\alpha$ & $x_1 x_5 x_3  x_2 x_4$\\
$\alpha$ & $x_1 x_5 x_4  x_2 x_3$
\end{tabular}
&
\begin{tabular}{ll}
$\alpha$ & $x_2 x_3 x_4 x_5 x_1$\\
$\alpha$ & $x_2 x_3 x_5 x_4 x_1$\\
$\alpha$ & $x_2 x_4 x_3 x_5 x_1$\\
$\alpha$ & $x_2 x_4 x_5 x_3 x_1$\\
$\alpha$ & $x_2 x_5 x_3 x_4 x_1$\\
$\alpha$ & $x_2 x_5 x_4 x_3 x_1$
\end{tabular}
&
\begin{tabular}{ll}
$\beta$ &  $x_3 x_4 x_5  x_2x_1$\\
$\beta$ & $x_3 x_5 x_4 x_2 x_1$\\
$\beta$ & $x_4 x_3 x_5 x_2 x_1$\\
$\beta$ & $x_4 x_5 x_3 x_2 x_1$\\
$\beta$ & $x_5 x_3 x_4 x_2 x_1$\\
$\beta$ & $x_5 x_4 x_3 x_2 x_1$
\end{tabular} 
\end{tabular} 
\end{center}
\normalsize
\medskip
\normalsize

Let $M = \frac{(m-3)!}{(m-k-1)!}$ and  $Q = \frac{(m-2)!}{(m-k-1)!}$.
\begin{lemma}\label{lemma-s-low1} 
$$S'_k(x_1) = S'_k(x_2) = \alpha (m-2) s'_1 M$$
$and$ $for$ $i \geq 3$,
$S'_k(x_i) = 2\alpha (s'_2 + \ldots + s'_k)  M + \beta (m-k-1) (s_1' + \ldots + s_k') M$
\end{lemma}


\begin{proof}
In $P_k$, $x_1$ and $x_2$ appear in top position in a number of votes equal to $\alpha$ times the number of different permutations (ordered lists) of $(k-1)$ candidates out of $(m-2)$, i.e. $\alpha\frac{(m-2)!}{(m-k-1)!}$ times. Thus $S_k'(x_1) = S_k'(x_2) = \alpha \frac{(m-2)!}{(m-k-1)!}s'_1$. 
For similar reasons, for each $i \geq 3$,
$$\begin{array}{lll}
S'_k(x_i) & = &  2 \alpha \frac{(m-3)!}{(m-k-1)!} (s'_2+\cdots+ s'_k)
 +   \beta \frac{(m-3)!}{(m-k-2)!} (s'_1 +\cdots+ s'_k).
\qedhere
\end{array} $$
\end{proof}

As a consequence, all candidates have the same score in $P_k$ if and only if
$$\frac{\beta}{\alpha} = \frac{(m-2)s'_1 - 2(s'_2 + \ldots + s'_k)}{(m-k-1)(s'_1 + \ldots + s'_k)}$$

We fix $\alpha$ and $\beta$ such that this equality holds.
Thanks to the tie-breaking priority, the winner in $P_k$ is $x_1$. In $P$, the winner is $x_2$ and the scores of $x_1$ and $x_2$ are as follows:

\begin{lemma}\label{lemma-borda-low3}
$$\begin{array}{lll}
S(x_1) & = & Q \alpha s_1\\
S(x_2) & = &  Q \alpha s_1 +Q \alpha s_{k+1} + Q (m-k-1) \beta s_{k+1}
\end{array}$$
\end{lemma}

\begin{proof}
$x_1$ appears at the top of $\frac{(m-2)!}{(m-k-1)!} \alpha$ votes and at the bottom of all others, hence $S(x_1)= Q \alpha s_1$.
$x_2$ appears  $\alpha\frac{(m-2)!}{(m-k-1)!}$ times top position, and in position $(k+1)$ in the remaining votes, {\em i.e.}, $\alpha \frac{(m-2)!}{(m-k-1)!}+ \beta \frac{(m-2)!}{(m-k-2)!}$. 
Thus 
$$\begin{array}{lll}
S(x_2) & = &  \alpha \frac{(m-2)!}{(m-k-1)!}(s_1 + s_{k+1}) + \beta \frac{(m-2)!}{(m-k-2)!}s_{k+1} \qedhere
\end{array}$$
\end{proof}

\begin{lemma}\label{lower}
$R(f^s,f^{s'}_{k},m,k) \geq 1 - \frac{s_{k+1}}{s_1} + \frac{s_{k+1}}{s_1}\frac{m s'_1}{s'_1 + \ldots + s'_k}$
\end{lemma}

\begin{proof}
From Lemma \ref{lemma-borda-low3} we get 
$\frac{S(x_2)}{S(x_1)} \geq 1 + \frac{s_{k+1}}{s_1} + (m-k-1)\frac{s_{k+1}}{s_1} \frac{\beta}{\alpha}.$


Finally, using the expression of $\frac{\beta}{\alpha}$ we get

$$\begin{array}{ll} & \frac{S(x_2)}{S(x_1)} 
\geq 1 + \frac{s_{k+1}}{s_1} + (m-k-1)\frac{s_{k+1}}{s_1}  \frac{(m-2)s'_1 - 2(s'_2 + \ldots + s'_k)}{(m-k-1)(s'_1 + \ldots + s'_k)}
\end{array}$$

From this we conclude:
$$\begin{array}{lll}   
R(f^s,f^{s'}_{k},m,k)& \geq & 1 + \frac{s_{k+1}}{s_1} + \frac{s_{k+1}}{s_1}  \frac{(m-2)s'_1 - 2(s'_2 + \ldots + s'_k)}{s'_1 + \ldots + s'_k}\\
&\geq & 1 + \frac{s_{k+1}}{s_1} + \frac{s_{k+1}}{s_1} \frac{(m-2)s'_1 + 2 s'_1 - 2(s'_1 + \ldots + s'_k)}{s'_1 + \ldots + s'_k}\\
&\geq & 1 + \frac{s_{k+1}}{s_1} + \frac{s_{k+1}}{s_1}\left( \frac{m s'_1}{s'_1 + \ldots + s'_k}-2\right)\\
&\geq & 1 - \frac{s_{k+1}}{s_1} + \frac{s_{k+1}}{s_1}\frac{m s'_1}{s'_1 + \ldots + s'_k} \qedhere
\end{array}$$ 
\end{proof}

Putting Lemmas \ref{upper} and \ref{lower} together we get 

\begin{Proposition}\label{propbounds}
$$\begin{array}{lll} 
1 - \frac{s_{k+1}}{s_1} + \frac{s_{k+1}}{s_1}\frac{m s'_1}{s'_1 + \ldots + s'_k} &
 \leq R(f^s,f^{s'}_{k},m,k) &
 \leq  1 - \frac{s_{k+1}}{s'_1} + \left(1+\frac{s^*}{s'_1}\right) \frac{m s_{k+1}}{s_1' + \ldots + s'_k}
\end{array}$$
\end{Proposition}

Note that the lower and upper bound coincide when $s^* = 0$, giving a tight worst-case approximation ratio for this class of approximations. This is however not guaranteed when $s^* > 0$ (the reason being that the pathological profile used in the proof of Lemma \ref{upper} may not be the worst).  Moreover, when $s^* = 0$, our (lower and upper) bound coincides with the optimal ratio given in \cite{bentert2019comparing} (Theorem 1).\footnote{Note that the ratios in our paper are the inverse of the ratios in \cite{bentert2019comparing}. That is, the inverse of the ratio given in Theorem 1 of \cite{bentert2019comparing} coincides with our ratio for $s^* = 0$.} Since the ratio in  \cite{bentert2019comparing} is shown to be the best possible ratio, this show that taking $s^* = 0$ gives a optimal top-$k$ approximation of a positional scoring rule.\footnote{Interestingly, \cite{bentert2019comparing}  give another optimal rule (thus with same worst-case ratio), which is much more complex, and which is not a top-$k$ PSR. Comparing the average ratio of both rules is left for further study.}\smallskip

In  particular:


\begin{itemize}
\item for  $Borda^0_k$ ($s_i = m-i, s^* = 0)$, the lower and upper bounds coincide and are equal to $\frac{k}{m-1} + \frac{2m(m-k-1)}{k(2m-k-1)}$.
\item for $Borda^{av}_k$ ($s_i = m-i, s^* = \nicefrac{m-k-1}{2}$),
the lower bound is $1-\frac{m-k-1}{m-1}+\frac{(m-k-1)(m+k-1)}{k(m-1)}$ and the upper bound is
$\frac{k(3k-m+1)+4(m-k-1)(m-1)}{k(m+k-1)}$.
\item for $Harmonic^0_k$ ($s_i = \nicefrac1i, s^* = 0)$, the lower and upper bounds are equal to $\frac{k}{k+1} + \frac{m}{(k+1)(1+\frac{1}{2}\dots+\frac{1}{k})}$.
\end{itemize}

Also, note that for $k'$-approval with $k' > k$ and $s^* = 0$, the (exact) worst-case ratio $\frac{m}{k}$ does not depend on $k'$.
As a corollary, we get the following order of magnitudes when $m$ grows:
\begin{itemize}
\item $R(Borda,Borda_{k}^{0},m,k) = \Theta\left(\frac{m}{k}\right)$.
\item $R(Borda,Borda_{k}^{av},m,k) = \Theta\left(\frac{m}{k}\right)$.
\item $R(Harmonic,Harmonic_{k}^{0},m,k) = \Theta\left(\frac{m}{k \log k}\right)$.
\end{itemize}

\paragraph{\bf Maximin }~~~~~~~~~~~~~~~~~~~~~~~~

Let $Maximin$ be the Maximin rule with tie-breaking priority $x_{1}  \ldots  x_{m}$, and $Maximin_k$ be the $k$-truncated version of the Maximin rule with the same tie-breaking priority order. 
Let $S_{Mm}(x_{2},P)$ and $S_{Mm}(x_{1},P_k)$ be the Maximin scores of $x_2$ and $x_{1}$ for $P$ and $P_k$, respectively, with $S_{Mm}(x_2,P) = \min_{y \neq x_2}N_{P}(x_2,y)$ and similarly for $P_k$. Let $P$ be a profile, and let $x_{1}=Maximin_k(P_k)$ and $x_{2} = Maximin(P)$. All candidates have the same Maximin score in $P_k$, therefore, by tie-breaking priority, $Maximin_k(P_k) = x_{1}$. 

\begin{lemma}\label{lemma-maximin-up1} 
$R(Maximin,Maximin_k,m,k) \leq  m - k + 1$.
\end{lemma}

\begin{proof}
Because $x_{1} = Maximin_k(P_k)$, we must have $S_{Mm}(x_{1},P_k) \geq 1$ (otherwise we would have $S_{Mm}(x_{1},P_k) \geq 0$, meaning that $x_1$ does not belong to any top-$k$ ballot, and in this case we cannot have $x_{1} = Maximin_k(P_k)$). Now, $S_{Mm}(x_{2},P) \leq S_{Mm}(x_{2},P_k) + (m-k) \leq  S_{Mm}(x_{1},P_k) + (m-k)$, therefore,
$$\begin{array}[b]{lll}
\frac{S_{Mm}(x_{2},P)}{S_{Mm}(x_{1},P)} & \leq & \frac{S_{Mm}(x_{1},P_k) + (m-k)}{S_{Mm}(x_{1},P_k)}\\
& \leq  & m - k + 1 \qedhere
\end{array}$$
\end{proof}

\begin{lemma}\label{lemma-maximin-l1} 
$R(Maximin,Maximin_k,m,k) \geq  m-k$.
\end{lemma}

\begin{proof}
We consider the cyclic profile $Cyc$:

\small
\begin{center}
\begin{tabular}{c|c}
$Cyc$ & $P$ $(m = 5, k = 2)$ \\ \hline 
\begin{tabular}{lllll}

$x_{1}$ & $x_{2}$& \ldots &$m-1$ & $m$ \\ 
$x_{2}$ & $x_{3}$&\ldots &$m$ &$x_{1}$ \\ 
$x_{3}$ & $x_{4}$ &\ldots &$x_{1}$ & $x_{2}$ \\ 
\ldots&\ldots& \ldots&\ldots\\ 
$m$ &$x_{1}$ & \ldots &$m-2$ &$m-1$ 

\end{tabular} 
&
\begin{tabular}{lllll}
$x_{1}$ & $x_{2}$ & $x_{3}$ & $x_{4}$ & $x_{5}$ \\
$x_{2}$ & $x_{3}$ & $x_{4}$ & $x_{5}$ & $x_{1}$\\
$x_{3}$ & $x_{4}$ & $x_{2}$ & $x_{5}$ & $x_{1}$\\
$x_{4}$ & $x_{5}$ & $x_{2}$ & $x_{3}$ & $x_{1}$\\
$x_{5}$ & $x_{1}$ & $x_{2}$ & $x_{3}$ & $x_{4}$ 
\end{tabular} 
\end{tabular} 
\end{center}
\normalsize

Now, let $P$ be obtained from $Cyc$ by the following operations for every vote in $Cyc$:

\begin{itemize}
\item if $x_{1}$ is not in the top $k$ positions in the vote, we move it to the last position (and move all candidates who were below $x_{1}$ one position upward)
\item if $x_{2}$ is not in the top $k$ positions in the vote, we move it to the $(k+1)^{th}$ position  (and move all candidates who were between position $k+1$ and 2's position one position downward).
\end{itemize}
For instance, for $m = 5$, $k = 2$, we get the profile $P$ above.

$Maximin(P) = x_{2}$, and the Maximin scores of $x_{1}$ and $x_{2}$ in $P$ are:
$$  S_{Mm}(x_{1},P) = 1 \mbox{ and } S_{Mm}(x_2,P) = m-k.$$
Hence  $\frac{S_{Mm}(x_2,P)}{S_{Mm}(x_{1},P)} = m-k$.\qedhere
\end{proof}

\begin{Proposition}
$m-k \leq R(Maximin,Maximin_k,m,k) \leq m-k+1.$
\end{Proposition}

This worst-case ratio is quite bad, except if $k$ is close to $m$. However, arguably, the maximin score makes less sense {\em per se} (i.e., as a measure of social welfare) than a positional score such as the Borda count.
Moreover, for maximin rule the obtained lower bound (Lemma \ref{lemma-maximin-l1}) matches the one given by Bentert and Skowron \cite{bentert2019comparing} (Section 4.3) which means that our top-k approximation of maximin is optimal.
\paragraph{\bf Copeland }~~~~~~~~~~~~~~

Again, for the Copeland rule, the ratio makes less sense, because the Copeland score is less meaningful as a measure of social welfare.\footnote{Moreover, there are several ways of defining the Copeland score, all leading to the same rule. However, 
this has no impact on the negative result below, as long as a Condorcet loser has score 0.} Still, for the sake of completeness we give the following result:

\begin{Proposition}
$R(Copeland,Copeland_k,m,k) = \infty$.
\end{Proposition}

\begin{proof}
Let $P$ be the following profile:
\begin{itemize}
\item $P_k$ contains two votes $x_1 x_2 \ldots x_k$, and one vote $L$ for each ordered list of $k$ candidates among $m$.
\item $P$ is obtained by completing $P_k$ by adding $x_1$ (resp. $x_2$) in last position (resp. in position $k+1$) when it is not in the $top$-$k$ positions.
\end{itemize}

In $P_k$, the winner for $Copeland_k$ is $x_1$. In $P$, the Copeland winner is $x_2$. Now, with respect to $P$, the Copeland score of $x_1$ (resp. $x_2$) is $0$ (resp. $m-1$), hence the result. \qedhere
\end{proof}

%

The obtained worst-case bounds are rather negative: very negative for Copeland and maximin, less so for Borda, and even less so for Harmonic.\footnote{As Ranked Pairs is not based on scores, it was not studied here.} However, the maximin and Copeland scores make less sense as a measure of social welfare than positional scores. 

Now, we may wonder whether these worst cases do occur frequently in practice or if they correspond to rare pathological profiles. The next two subsections show that the latter is the case.
\subsection{Average Case Evaluation}\label{AppRatioPractice}

We present the evaluation of the approximation ratio using data generated from Mallows $\phi$ model. For each experiment, we draw 10000 random profiles, with $m = 7$, $n=15$, and let $\phi$ vary.

Fig.~\ref{RatioMallows} shows results reflecting the approximation ratio for truncated rules when using Mallows model. 
Our results suggest that, 
in practice, results are much better than in the worst case where best results are obtained by Harmonic, followed by Borda and finally Maximin.  

\begin{figure}[h]
\centering
    \includegraphics[width=11cm]{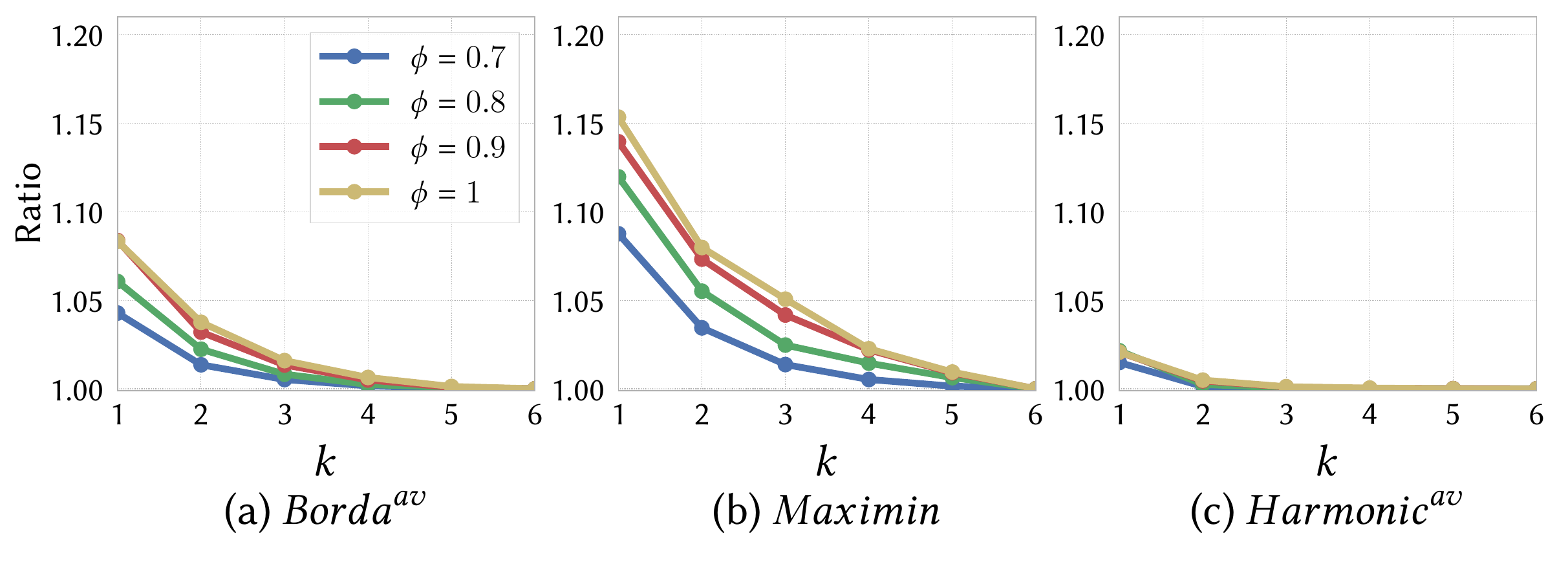}
	\caption{Mallows model: approximation ratio when $n=15$, $m=7$ and varying $\phi$. } 
\label{RatioMallows}
\end{figure}

\subsection{Real Data Sets}
Again we consider 2002 Dublin North data ($m=12,n=3662$) with samples of $n^*$ voters among $n$ ($n^{*}<n$) where $n^{*}=\{15,100\}$. In each experiment 1000 random profiles are constructed with $n^*$ voters; then we consider the top-k ballots obtained from these profiles with $k=\{1,\dots,m-1\}$.
Again, the results are very positive.
\vspace{0.5cm}
\begin{figure}[h]
\centering
    \includegraphics[width=10cm]{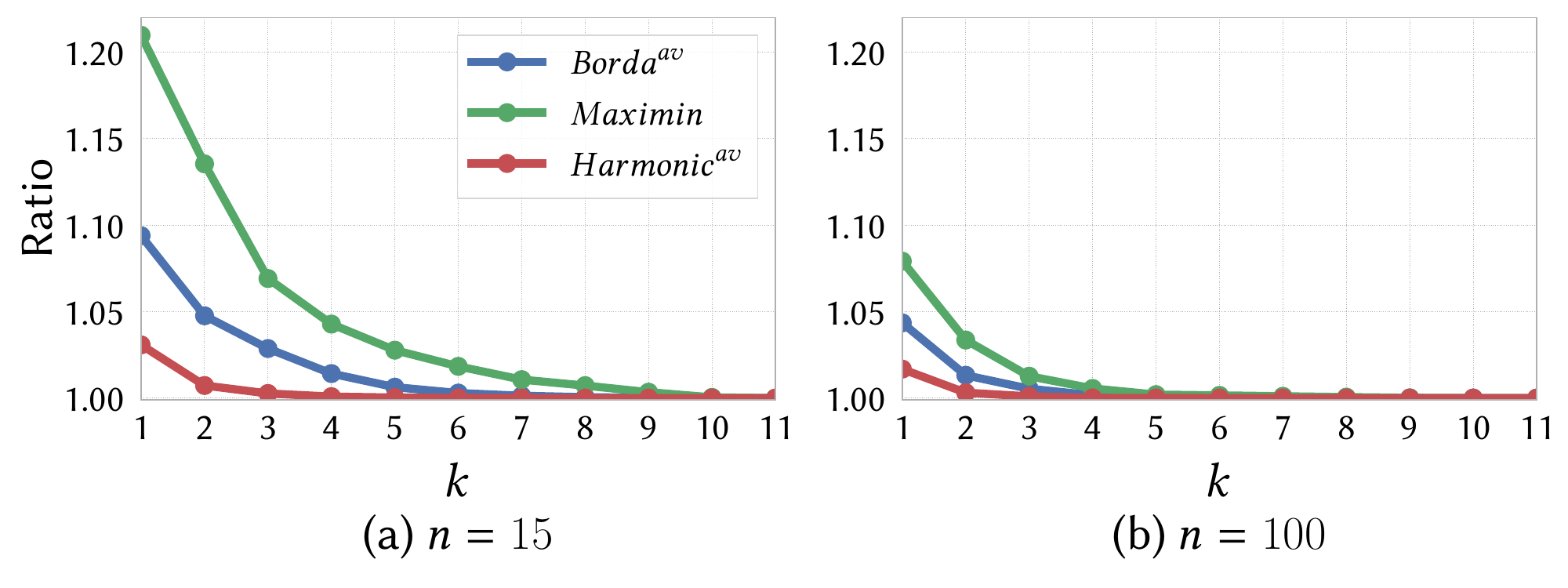}
	\caption{Approximation ratio with Dublin North data set.} 
\label{RatioReal}
\end{figure}
\section{Conclusion}
In this paper we have considered $k$-truncated approximations of rules which take only $top$-$k$ ballots as input where we have considered two measures of the quality of the approximation: the probability of selecting the same winner as the original rule, and the score ratio.
For the former, our empirical study show that a very small $k$ suffices.
For the latter, while the theoretical bounds are, at best; moderately encouraging, our experiments show that in practice the approximation ratio is much better than in the worst case: our results 
suggest that a very small value of $k$ works very well in practice.
Many issues remain open. Especially, it would be interesting to consider top-$k$ approximations as voting rules on their own, and to study their normative properties.

\bibliographystyle{splncs04}

\end{document}